\documentclass[letterpaper,10pt,conference]{ieeeconf}

\IEEEoverridecommandlockouts  

\usepackage{amsmath}
\usepackage{amssymb}
\usepackage{amsthm}
\usepackage{mathtools}
\usepackage{derivative}
\usepackage{xcolor}
\usepackage{bm}

\usepackage[noadjust]{cite}
\usepackage{url}

\usepackage{graphicx}
\usepackage[export]{adjustbox} 
\graphicspath{{figures/}}

\newtheorem{theorem}{Theorem}
\newtheorem{lemma}{Lemma}
\newtheorem{proposition}{Proposition}
\newtheorem{corollary}{Corollary}

\theoremstyle{definition}
\newtheorem{definition}{Definition}

\newtheorem{remark}{Remark}

\newcommand{\argmin}{\operatornamewithlimits{arg\,min}}

\newcommand{\R}{\mathbb{R}}

\renewcommand{\bf}{\mathbf{f}} 

\usepackage{xcolor}
\definecolor{myblue}{RGB}{49, 114, 174}
\definecolor{myred}{rgb}{0.796, 0.235, 0.2}
\definecolor{mygreen}{rgb}{0.22, 0.596, 0.149}
\definecolor{mypurple}{rgb}{0.584,0.345,0.698}

\usepackage{blindtext}

\usepackage{amsfonts}
\usepackage{hyperref}
\usepackage{comment}

\newcommand{\Cset}{\mathcal{C}}

\title{\textbf{
Stability Margins of CBF-QP Safety Filters: Analysis and Synthesis
}}

\author{Shima Sadat Mousavi, Pol Mestres, and Aaron D. Ames%
\thanks{The authors are with the Department of Mechanical and Civil Engineering,
California Institute of Technology, Pasadena, CA 91125, USA.
\texttt{\{smousavi,mestres,ames\}@caltech.edu}.
This work was supported by The Boeing Company.}%
}

\begin{document}
\maketitle
\thispagestyle{empty}
\pagestyle{empty}


\begin{abstract}
Control barrier function (CBF)-QP safety filters enforce safety by minimally modifying a nominal controller. While prior work has mainly addressed robustness of \emph{safety} under uncertainty, robustness of the resulting closed-loop \emph{stability} is much less understood. This issue is important because once the safety filter becomes active, it modifies the nominal dynamics and can reduce stability margins or even destabilize the system, despite preserving safety. For linear systems with a single affine safety constraint, we show that the active-mode dynamics admit an exact scalar loop representation, leading to a classical robust-control interpretation in terms of gain, phase, and delay margins. This viewpoint yields exact stability-margin characterizations and tractable linear matrix inequality (LMI)-based certificates and synthesis conditions for controllers with certified robustness guarantees. Numerical examples illustrate the proposed analysis and the enlargement of certified stability margins for safety-filtered systems.
\end{abstract}



\section{Introduction}


Safety-critical systems such as autonomous cars, legged robots, aerospace vehicles, and power networks must satisfy state constraints without sacrificing closed-loop performance. Control-theoretic tools for this purpose include control barrier functions (CBFs), model predictive control, Hamilton--Jacobi reachability, and reference governors \cite{ADA-XX-JWG-PT:17,ADA-SC-ME-GN-KS-PT:19,JBR-DQM-MMD:17,SB-MC-SH-CJT:17,EG-SDC-IK:17}. Among these, CBF-based safety filters are especially appealing because they can be wrapped around a pre-existing controller and intervene only when needed, typically through a small online quadratic program \cite{ADA-XX-JWG-PT:17,ADA-SC-ME-GN-KS-PT:19,mousavi2025vertices,mestres2025explicit}. This modular structure has enabled safety-critical control in adaptive cruise control, bipedal walking, and aircraft applications \cite{ADA-JWG-PT:14,SCH-XX-ADA:15,OS-ZS-MM-JG-KR-NR-CF:24,EL-KAW:24}. However, because the filter acts as a supervisory correction layer around a nominal stabilizing controller, implementation error, phase lag, or delay can affect closed-loop stability when the filter becomes active.

A substantial literature has therefore studied robustness of CBFs from the viewpoint of \emph{safety preservation} under uncertainty. Existing results address perturbations to the dynamics, input disturbances and input-to-state safety, measurement uncertainty, sampled-data implementations, and stochastic uncertainty \cite{XX-PT-JWG-ADA:15,KolathayaAmes2019,CosnerSingletaryTaylorMolnarBoumanAmes2021,OrugantiNaghizadehAhmed2024,CosnerCulbertsonTaylorAmes2023}. These works are essential because they show how forward invariance can be maintained despite modeling error, estimation error, and implementation imperfections. However, the main certificate in this literature is robustness of \emph{safety}: it does not quantify, in the classical robust-control sense, how much closed-loop \emph{stability} robustness is lost or preserved once the safety filter modifies the nominal stabilizing controller.

At the same time, a different line of work has shown that CBF-QP safety filters can induce unwanted dynamical effects such as undesired equilibria, limit cycles, and unbounded trajectories \cite{MFR-APA-PT:21,XT-DVD:24,PM-YC-EDA-JC:25-jnls1}. The recent paper~\cite{ChoiTomlinSastrySreenath2025} further highlights that, when treating the CBF as an output, the internal dynamics can lead to instability even when the barrier condition is enforced. For linear systems with an affine safety constraint, the recent work~\cite{PM-SSM-ADA:26} provides simple conditions ensuring global exponential stability of the safety-filtered system. In parallel, classical robust control studies the effect of uncertainty on closed-loop stability through loop representations, frequency-domain tests, and gain, phase, and delay margins \cite{ZhouDoyleGlover1996,BoydElGhaouiFeronBalakrishnan1994,EL-KAW:24}. What remains missing is a bridge between these viewpoints: a framework for certified robustness of the \emph{stability} of safety-filtered dynamics, rather than only robustness of their safety guarantees.

The CBF-QP safety filter can be viewed as a supervisory correction layer wrapped around a nominal controller. In this architecture, the nominal controller defines the baseline closed loop, while the safety filter adds a corrective action only when the barrier constraint is active. This motivates modeling uncertainty in the correction channel of the safety filter, since such perturbations affect the closed loop only during safety intervention. Although robustness of the full safety-filtered system is, in principle, determined by both the inactive and active modes, the inactive mode is exactly the nominal closed loop, and we assume that the baseline controller has already been designed to provide the desired robustness margin there. Accordingly, the relevant robustness degradation arises when the safety filter intervenes, and the active mode becomes the focus of the present analysis.

This paper develops such a framework for linear systems with an affine safety constraint. First, building on \cite{PM-SSM-ADA:26}, we show that the active-mode dynamics admit an exact scalar loop representation around the nominal closed loop (Prop.~\ref{prop:loop_transfer}). This leads to an exact frequency-domain characterization of active-mode robustness under scalar loop perturbations (Thm.~\ref{thm:active_loop_robustness}), including gain, phase, and delay margin interpretations and a direct safety observation for correction-channel gain uncertainty (Cor.~\ref{cor:gain_interval_freq} and Prop.~\ref{prop:safety_kappa}). Second, we derive tractable linear matrix inequality (LMI)-based certificates for certified gain, phase, and delay margins (Prop.~\ref{prop:lmi_gain_interval}, Prop.~\ref{prop:lmi_phase_bound}, and Cor.~\ref{cor:delay_from_phase}). Third, we provide convex synthesis conditions for selecting the nominal feedback gain to achieve a prescribed certified gain margin (Lem.~\ref{lem:Atilde_affine_in_K} and Prop.~\ref{prop:gain_interval_synthesis}). 

\section{Problem Setup and Safety-Filtered Dynamics}\label{sec:setup}

We consider linear systems with a single affine safety constraint and the associated CBF--QP safety filter. We first recall the closed-form controller and the induced two-mode piecewise-affine dynamics from \cite{PM-SSM-ADA:26}, which form the basis of the robustness analysis developed here.

\subsection{System and Safety Filter}

Consider the linear system
\begin{equation*}
\dot x = A x + B u ,
\label{eq:sys}
\end{equation*}
where $x\in\R^n$ and $u\in\R^m$. Let
\(
u_{\rm nom}(x) = -Kx
\)
be a nominal state-feedback controller, and define $A_0 := A-BK$.
Throughout the paper, we assume that $A_0$ is Hurwitz.
Safety is specified by the affine constraint $h(x) := c^\top x + d \ge 0$,
with associated safe set 
\begin{equation}
    \Cset := \{x\in\R^n : h(x)\ge 0\}.
    \label{eq:safe_set}
    \end{equation}
    We assume $d>0$,
so that the origin lies in the interior of $\Cset$,  and
\(c^\top B\neq 0\), so that the input enters the barrier derivative.
To enforce safety, we use the exponential CBF condition
\begin{equation}
c^\top(Ax+Bu) + \alpha h(x) \ge 0 ,
\label{eq:cbf}
\end{equation}
where $\alpha>0$. The safety-filtered input is defined by the quadratic program
\begin{equation}
\begin{aligned}
u^\star(x)=\argmin_{u\in\R^m}\;&
\frac{1}{2}(u+Kx)^\top G (u+Kx) \\
\text{s.t.}\;&
c^\top(Ax+Bu)+\alpha h(x)\ge 0 ,
\end{aligned}
\label{eq:qp}
\end{equation}
where $G=G^\top\succ 0$.


For every \(x\), \eqref{eq:qp} admits a unique optimizer, since the objective
is strictly convex and the constraint is affine. It therefore returns the
closest safe input to \(u_{\rm nom}(x)\) in the \(G\)-weighted norm. We next
recall the resulting closed-form controller and associated two-mode dynamics
from \cite{PM-SSM-ADA:26}.

\subsection{Closed-Form Controller and Two-Mode Dynamics}

We next recall from \cite{PM-SSM-ADA:26} the explicit CBF--QP solution and the induced two-mode closed-loop dynamics. The resulting safety-filtered closed loop is continuous and piecewise affine, with an active mode that differs from the nominal mode by a rank-one term; this structure is central to the robustness analysis below. Define \(a:=B^\top c\), \(\beta:=a^\top G^{-1}a\), and \(g(x):=c^\top (A_0+\alpha I)x+\alpha d\). The function \(g(x)\) determines whether the nominal input satisfies the CBF constraint. Accordingly, define the inactive and active regions
%
\begin{equation}
\mathcal{R}_+ := \{x\in\R^n : g(x)\ge 0\},
\qquad
\mathcal{R}_- := \R^n\setminus\mathcal{R}_{+}.
\label{eq:regions}
\end{equation}

Define
\begin{equation}
U := \frac{B G^{-1} a}{\beta},
\qquad
V^\top := c^\top (A_0+\alpha I),
\label{eq:UV}
\end{equation}
and $\tilde{A} = A_0 - UV^\top$, $\tilde{b} = -\alpha d U$.

The next result recalls the explicit form of the safety filter and the induced
two-mode closed loop. The key point is that, in the active region, the safety
filter adds a correction in a fixed direction, which leads to an affine mode
whose linear part is a rank-one modification of the nominal closed loop.

\begin{proposition}[{\cite{PM-SSM-ADA:26}}]
\label{prop:twomode}
The optimizer of \eqref{eq:qp} is
\begin{equation}
u^\star(x)=
\begin{cases}
-Kx, & x\in \mathcal{R}_+, \\[1mm]
-Kx-\dfrac{g(x)}{\beta}\,G^{-1}a, & x\in \mathcal{R}_- ,
\end{cases}
\label{eq:ustar}
\end{equation}
and the corresponding closed-loop dynamics are
\begin{equation}
\dot x =
\begin{cases}
A_0 x, & x\in \mathcal{R}_+, \\[1mm]
\tilde A x + \tilde b, & x\in \mathcal{R}_- .
\end{cases}
\label{eq:twomode}
\end{equation}
\end{proposition}

\begin{remark}
\label{rem:corr_channel}
In the active region \(\mathcal{R}_-\) in \eqref{eq:regions}, the safety-filtered input can be written as
\begin{equation}
u^\star(x)=u_{\rm nom}(x)+u_{\rm corr}(x),
\qquad
u_{\rm nom}(x)=-Kx,
\label{eq:nom_corr_split}
\end{equation}
with correction term
\begin{equation}
u_{\rm corr}(x):=-\frac{g(x)}{\beta}\,G^{-1}a.
\label{eq:u_corr}
\end{equation}
Thus, in the active region, the safety filter modifies the nominal controller through the additive correction term \(u_{\rm corr}(x)\).
\end{remark}


Equation \eqref{eq:twomode} shows that the safety-filtered closed loop is a continuous two-mode piecewise-affine system. Its active-mode matrix \(\tilde A\) differs from the nominal closed-loop matrix \(A_0\) by the rank-one term \(-UV^\top\), a structure underlying the scalar loop representation and robustness analysis developed next.
We also recall the following stability result from \cite{PM-SSM-ADA:26}, which will be used later to translate mode-wise stability into stability of the full closed loop~\eqref{eq:twomode}.

\begin{proposition}[{\cite{PM-SSM-ADA:26}}]
\label{prop:ges_recall}
If \(\tilde A\) is Hurwitz, then the origin is globally exponentially stable
(GES) for \eqref{eq:twomode}.
\end{proposition}


The goal of the paper is to characterize robustness margins of the active-mode loop induced by the safety filter and to derive design conditions on \(K\) that enlarge these margins while preserving stability of the closed loop~\eqref{eq:twomode}. Although robustness of \eqref{eq:twomode} depends, in principle, on both modes, the inactive mode is exactly the nominal closed loop \(A_0\), whose robustness margin is assumed satisfactory by design. Since the safety filter modifies the dynamics only through the correction term \(u_{\rm corr}(x)\) when active, the relevant robustness degradation arises in the active mode.

\section{Active-Mode Robustness}
\label{sec:active_mode_robustness}

This section develops the robustness viewpoint of the paper. Since the safety
filter modifies the dynamics only in the active region, we focus on the active
mode of \eqref{eq:twomode}. We show that its rank-one structure admits an
equivalent scalar loop representation around the nominal closed loop. This
reduction allows gain, phase, and delay margins to be characterized using classical loop-based tools.

\subsection{Scalar Loop Representation of the Active Mode}
\label{subsec:loop_representation}

We begin by showing that the active mode admits an exact scalar feedback
representation around the nominal closed loop. The key idea is that the
rank-one correction induced by the safety filter acts through a single input
direction and a single output channel. This observation is the basis for the
robustness analysis developed in the remainder of the section.

Consider the active-mode linear part
\begin{equation}
\dot z = \tilde A z,
\label{eq:active_linear}
\end{equation}
where $z$ denotes the state in coordinates centered at the active-mode
equilibrium whenever the latter exists.

The next result shows that the active mode admits an exact scalar loop representation around the nominal closed loop.

\begin{proposition}
\label{prop:loop_transfer}
The active mode \(\dot z=\tilde A z\) can be realized as the closed loop of a scalar feedback interconnection whose open-loop transfer function is
\begin{equation}
L(s)=V^\top (sI-A_0)^{-1}U.
\label{eq:return_ratio}
\end{equation}
\end{proposition}

\begin{proof}
Consider the auxiliary SISO system
\begin{equation}
\dot z = A_0 z + Uw, \qquad y = V^\top z.
\label{eq:aux_iosys_proof}
\end{equation}
Its transfer function from the scalar input \(w\) to the scalar output \(y\) is written as \eqref{eq:return_ratio}.
Now impose the static negative-feedback interconnection
\(
w=-y.
\)
Since \(y=V^\top z\), this gives \(w=-V^\top z\), and therefore
\(
\dot z
=
A_0 z + U(-V^\top z)
=
(A_0-UV^\top)z.
\)
Recalling from that
\(
\tilde A = A_0-UV^\top,
\)
we conclude that the resulting closed-loop dynamics are
\(
\dot z=\tilde A z,
\)
which is exactly the active-mode linear dynamics.
\end{proof}

Prop.~\ref{prop:loop_transfer} shows that the active mode can be realized
as a scalar feedback loop around the nominal closed loop. 
Inspired by this, we model
uncertainty directly in this feedback loop by considering the perturbed feedback interconnection
\begin{equation}
\dot z = A_0 z + Uw, \qquad
y = V^\top z, \qquad
w = -\Delta y,
\label{eq:active_loop_delta}
\end{equation}
where \(\Delta\in\mathbb{C}\) denotes a scalar loop perturbation. The corresponding
perturbed active-mode matrix is
\begin{equation}
\tilde A(\Delta):=A_0-\Delta UV^\top.
\label{eq:AtildeDelta}
\end{equation}
This model captures, at an abstract level, uncertainty affecting the
additional feedback action induced by the safety filter, including
correction-channel gain error, approximation error, unmodeled dynamics, and delay.

A particularly relevant structured case is multiplicative uncertainty in the safety-correction term. The motivation is that the CBF-QP safety filter acts as a supervisory layer around an already designed nominal controller: it does not redesign the baseline loop, but modifies it only through the additive correction term \(u_{\rm corr}(x)\) when the barrier constraint is active. Thus, uncertainty associated specifically with the safety layer is naturally modeled in that correction channel. A multiplicative perturbation is especially natural when the correction is computed in the correct direction but applied with uncertain effective gain, for instance due to gain mismatch, limited control authority, actuator scaling, blending with other channels, or approximate QP implementation. Moreover, since Prop.~\ref{prop:loop_transfer} shows that the active mode admits an exact scalar feedback realization, this perturbation is also the natural structured uncertainty model for classical gain, phase, and delay margin analysis. Accordingly, using \eqref{eq:nom_corr_split}--\eqref{eq:u_corr}, we replace \(u_{\rm corr}\) with \(\Delta u_{\rm corr}\). In the active region, this yields \eqref{eq:AtildeDelta}. Thus, the loop perturbation model above exactly captures multiplicative uncertainty in the correction channel.

We now state the main result of the paper characterizing robustness of the active mode dynamics.
The key point is that
stability of the perturbed active mode is exactly equivalent to a scalar
closed-loop condition. This converts active-mode robustness into a classical
loop-stability question.

\begin{theorem}
\label{thm:active_loop_robustness}
The matrix \(\tilde A(\Delta)\) in \eqref{eq:AtildeDelta} is Hurwitz if and only if the scalar closed-loop equation
\begin{equation}
1+\Delta L(s)=0,
\label{eq:scalar_loop_general}
\end{equation}
with \(L(s)\) defined in \eqref{eq:return_ratio}, has no solutions in \(\{\Re s\ge 0\}\). In that case, the origin is GES for the closed loop \eqref{eq:twomode}.
\end{theorem}

\begin{proof}
Since \(A_0\) is Hurwitz, \(sI-A_0\) is invertible for every \(s\) with
\(\Re s\ge 0\). Moreover, \(\tilde A(\Delta)=A_0-\Delta UV^\top\) is a rank-one
perturbation of \(A_0\). Applying the matrix determinant lemma we get
\(
\det(sI-\tilde A(\Delta))
=
\det(sI-A_0)\bigl(1+\Delta L(s)\bigr).
\)
Hence, $\det(sI-\tilde A(\Delta))\neq 0$ for all $s$ with $\Re s\ge 0$ if and
only if $1+\Delta L(s)\neq 0$ for all such $s$. This proves the equivalence.
Since \(A_0\) is Hurwitz by assumption and \(\tilde A(\Delta)\) is Hurwitz,
Prop.~\ref{prop:ges_recall} implies that the origin is GES.
\end{proof}

Thm.~\ref{thm:active_loop_robustness} reduces active-mode robustness to a
scalar loop-stability condition. In the next subsections, we specialize this
result to gain, phase, and delay perturbations and derive the associated
robustness margins.

\subsection{Classical Robustness Margins}
\label{subsec:gain_robustness}

We now specialize the uncertainty model~\eqref{eq:AtildeDelta} to real gain perturbations. In
this case, Thm.~\ref{thm:active_loop_robustness} characterizes exactly when
the perturbed active mode remains Hurwitz and therefore when GES of the origin for the closed loop is guaranteed.

For real gain perturbations, we set $\Delta = \kappa$, $\kappa > 0$, so that
\begin{equation}
\tilde A(\kappa):=A_0-\kappa UV^\top.
\label{eq:Atilde_kappa}
\end{equation}

This motivates the following definition, which captures the notion of gain margin for~\eqref{eq:twomode}.

\begin{definition}[Stability-preserving gain interval]
\label{def:gain_interval}
Assume 
$\tilde A = \tilde{A}(1)$ is Hurwitz and define
\begin{equation}
\mathcal K:=\{\kappa>0:\tilde A(\kappa)\ \text{is Hurwitz}\}.
\label{eq:Kappa_set}
\end{equation}
The stability-preserving gain interval is the connected component of
\(\mathcal K\) that contains \(\kappa=1\).
\end{definition}

\begin{remark}
In classical single-loop control, gain margin is often summarized by one or two
boundary values relative to the nominal gain. In the present setting, the
natural object is the full interval of gain perturbations for which the active
mode remains Hurwitz, and the origin is guaranteed to be GES for the closed loop~\eqref{eq:twomode}. 
\end{remark}

The next corollary gives the frequency-domain characterization of the boundary
of the interval in Definition~\ref{def:gain_interval}. Intuitively, loss of
admissibility occurs when the Nyquist curve of $L(j\omega)$ reaches the
negative real axis at the critical value $-1/\kappa$.

\begin{corollary}
\label{cor:gain_interval_freq}
Let \((\underline\kappa,\bar\kappa)\) denote the stability-preserving gain interval, and let \(L(s)\) be defined by \eqref{eq:return_ratio}. If \(\kappa_\star\in\{\underline\kappa,\bar\kappa\}\cap(0,\infty)\), then there exists \(\omega\ge 0\) such that
\begin{equation}
1+\kappa_\star L(j\omega)=0.
\label{eq:gain_boundary_loop}
\end{equation}
Equivalently, \(\kappa_\star=\frac{1}{|L(j\omega)|}\) for some \(\omega\ge 0\) such that \(L(j\omega)\in\mathbb R_{<0}\).
\end{corollary}

\begin{proof}
Let \(\kappa_\star\in\{\underline\kappa,\bar\kappa\}\cap(0,\infty)\) be a finite endpoint
of the stability-preserving gain interval. Since
\(\kappa\mapsto \tilde A(\kappa)\) is continuous and the set of Hurwitz matrices
is open, \(\tilde A(\kappa_\star)\) cannot be Hurwitz. Also, because
\(\kappa_\star\) is a boundary point of a connected stability interval,
\(\tilde A(\kappa_\star)\) cannot have an eigenvalue in \(\{\Re s>0\}\);
otherwise continuity of eigenvalues would imply instability for nearby gains
inside the interval, a contradiction. Hence \(\tilde A(\kappa_\star)\) has an
eigenvalue on the imaginary axis, say \(j\omega\) with \(\omega\ge 0\).
By Thm.~\ref{thm:active_loop_robustness},
\(
1+\kappa_\star L(j\omega)=0.
\)
Since \(\kappa_\star>0\), this implies \(L(j\omega)\in\mathbb R_{<0}\) and
\(
\kappa_\star=\frac{1}{|L(j\omega)|}.
\)
\end{proof}

For the special case of correction-channel scaling, one can also make a direct
safety statement. Indeed, in this case the perturbation modifies the implemented
control input itself, so its effect can be computed directly in the CBF
inequality. 

\begin{proposition}
\label{prop:safety_kappa}
Let \(\Cset\) be the safe set in \eqref{eq:safe_set}, \(\mathcal R_-\) the active region in \eqref{eq:regions}, and \(u_{\rm corr}(x)\) the correction term in \eqref{eq:u_corr}. If \(u_{\rm corr}(x)\) is scaled as
\[
u_{\rm corr}(x)\mapsto \kappa u_{\rm corr}(x), \qquad \kappa>0,
\]
then \(\Cset\) remains forward invariant for every \(\kappa\ge 1\). If \(0<\kappa<1\), then the CBF condition \eqref{eq:cbf} is violated for all \(x\in\mathcal R_-\).
\end{proposition}

\begin{proof}
From \eqref{eq:u_corr}, we have 
\(
a^\top u_{\rm corr}(x)
=
-\frac{g(x)}{\beta}\,a^\top G^{-1}a
=
-g(x).
\)
Now consider the scaled correction input
\(
u_\kappa(x)=u_{\rm nom}(x)+\kappa u_{\rm corr}(x).
\)
Since \(u_{\rm nom}(x)=-Kx\), using the definition of $g(x)$ gives
\(
c^\top(Ax+Bu_{\rm nom}(x))+\alpha h(x)=g(x).
\)
Therefore,
\begin{align}
&c^\top(Ax+Bu_\kappa(x))+\alpha h(x)
=
g(x)+\kappa\,a^\top u_{\rm corr}(x) \notag\\
&=
g(x)-\kappa g(x)
= (1-\kappa)g(x).
\label{eq:cbf_scaled}
\end{align}
If \(x\in\mathcal R_-\), then \(g(x)<0\). Thus, for every \(\kappa\ge 1\),
\(
(1-\kappa)g(x)\ge 0,
\)
so the CBF inequality \eqref{eq:cbf} is satisfied in the active region. In the
inactive region \(\mathcal R_+\), the nominal input is already feasible by
definition of \(g(x)\), and scaling the correction does not affect the applied
input there. Hence the CBF condition holds everywhere, so the standard CBF
argument implies that \(\Cset\) remains forward invariant.
On the other hand, if \(0<\kappa<1\) and \(x\in\mathcal R_-\), then
\(
(1-\kappa)g(x)<0.
\)
By \eqref{eq:cbf_scaled}, this means that the CBF condition \eqref{eq:cbf} is violated throughout \(\mathcal R_-\).
\end{proof}

Prop.~\ref{prop:safety_kappa} shows that, for correction-channel
uncertainty, preservation of safety and preservation of exponential stability
need not coincide. 
Indeed, the stability-preserving gain interval might not coincide with the interval $[1,\infty)$, which is the set where safety is preserved.

Beyond gain robustness, the same loop-based viewpoint yields the corresponding phase- and delay-robustness conditions for the active mode. In particular, if a constant phase shift \(e^{-j\phi}\) or a delay factor \(e^{-s\tau}\) is introduced in the active loop, then by Thm.~\ref{thm:active_loop_robustness} the perturbed active mode remains Hurwitz if and only if the scalar closed-loop equations
\begin{equation}
1+e^{-j\phi}L(s)=0,
\qquad
1+e^{-s\tau}L(s)=0
\label{eq:phase_delay_loop}
\end{equation}
have no solutions in \(\{\Re s\ge 0\}\), where \(\phi\in\R\) and \(\tau\ge 0\).
As in classical loop analysis, the
boundary of the phase- and delay-robustness ranges is determined by
gain-crossover frequencies \(\omega\ge 0\) such that $|L(j\omega)|=1$.
At such frequencies, the boundary phase shift and delay satisfy
\begin{equation}
\phi \equiv \pi+\arg L(j\omega)\;(\mathrm{mod}\;2\pi),
\label{eq:phase_boundary}
\end{equation}
and
\begin{equation}
\tau=\frac{\pi+\arg L(j\omega)+2\pi k}{\omega},
\qquad k\in\mathbb Z,
\label{eq:delay_boundary}
\end{equation}
respectively.

\begin{remark}
Unlike the gain-scaling case in Prop.~\ref{prop:safety_kappa}, phase and
delay perturbations do not admit a comparably simple pointwise safety test,
since they do not act as static real scalings of the corrective input.
Instead, the results above give exact phase- and delay-robustness conditions
for the active mode and, via Prop.~\ref{prop:ges_recall}, sufficient
conditions for global exponential stability of the full safety-filtered closed
loop.
\end{remark}

\section{LMI-Based Certification and Synthesis}
\label{sec:lmi_certificates}

Section~\ref{sec:active_mode_robustness} gives exact gain, phase, and delay
robustness conditions for the perturbed active mode through the scalar loop
transfer function \(L(s)\). Together with the stability of the nominal mode
\(A_0\), these conditions provide sufficient guarantees that the origin is GES for the closed loop~\eqref{eq:twomode}  on \(\mathbb{R}^n\). This section
develops tractable LMI-based sufficient conditions that give conservative inner
approximations of the corresponding robustness ranges and can be used directly
for synthesis. We first focus on the gain case, for which the perturbation
enters affinely, and then give certified phase- and delay-robustness bounds.

\subsection{LMI-Certified Gain Intervals}
\label{subsec:lmi_gain_intervals}

We begin with the gain-robustness case. Since the perturbed active mode
\(\tilde A(\kappa)\) depends affinely on the gain parameter \(\kappa\), common
quadratic Lyapunov certificates yield tractable sufficient conditions for
guaranteeing prescribed gain intervals. These same certificates can then be
used to compute conservative inner approximations of the exact
stability-preserving gain interval from Definition~\ref{def:gain_interval}.

For a prescribed interval
$[\underline\kappa,\bar\kappa]\subset\R_{>0}$, $\underline\kappa\le 1\le \bar\kappa$
we recall that
\begin{equation}
\tilde A(\kappa)=A_0-\kappa UV^\top.
\label{eq:Atilde_kappa_lmi}
\end{equation}

The next result gives a tractable sufficient condition for guaranteeing
stability over the entire interval $[\underline{\kappa},\bar{\kappa}]$.
The key point is
that \(\tilde A(\kappa)\) depends affinely on \(\kappa\), so the Lyapunov
inequality can be enforced on the whole interval by checking only its
endpoints.

\begin{proposition}
\label{prop:lmi_gain_interval}
Fix \([\underline\kappa,\bar\kappa]\subset\mathbb{R}_{>0}\), and let \(\tilde A(\kappa)\) be defined by \eqref{eq:Atilde_kappa_lmi}. Suppose there exists \(P=P^\top\succ 0\) such that
\begin{equation}
A_0^\top P+PA_0\prec 0,
\label{eq:lmi_A0}
\end{equation}
and
\begin{equation}
\tilde A(\underline\kappa)^\top P+P\tilde A(\underline\kappa)\prec 0, \qquad
\tilde A(\bar\kappa)^\top P+P\tilde A(\bar\kappa)\prec 0.
\label{eq:lmi_kappa_high}
\end{equation}
Then \(\tilde A(\kappa)\) is Hurwitz for all \(\kappa\in[\underline\kappa,\bar\kappa]\), and the origin is GES for the closed loop \eqref{eq:twomode} for every \(\kappa\in[\underline\kappa,\bar\kappa]\).
\end{proposition}

\begin{proof}
For every \(\kappa\in[\underline\kappa,\bar\kappa]\), there exists
\(\lambda\in[0,1]\) such that
$\kappa=\lambda \underline\kappa+(1-\lambda)\bar\kappa$.
Since \(\tilde A(\kappa)\) depends affinely on \(\kappa\),
\begin{equation}
\tilde A(\kappa)
=
\lambda \tilde A(\underline\kappa)
+(1-\lambda)\tilde A(\bar\kappa).
\label{eq:Atilde_convex}
\end{equation}
Hence,
\begin{align}
\tilde A(\kappa)^\top P+P\tilde A(\kappa)
&=
\lambda\bigl(\tilde A(\underline\kappa)^\top P
+P\tilde A(\underline\kappa)\bigr) \notag\\
&\quad
+(1-\lambda)\bigl(\tilde A(\bar\kappa)^\top P
+P\tilde A(\bar\kappa)\bigr).
\label{eq:lyap_convex}
\end{align}
By~\eqref{eq:lmi_kappa_high}, the right-hand side is
negative definite, so \(\tilde A(\kappa)\) is Hurwitz for every
\(\kappa\in[\underline\kappa,\bar\kappa]\). Since \(A_0\) is Hurwitz, Prop.~\ref{prop:ges_recall} implies that the origin is GES.
\end{proof}

Prop.~\ref{prop:lmi_gain_interval} gives a conservative but convex
certificate for a prescribed gain interval. We next use it to define a
computable inner approximation of the exact interval in Definition~\ref{def:gain_interval}.

\begin{definition}[LMI-certified gain interval]
\label{def:lmi_gain_interval}
Let \(\tilde A(\kappa)\) be defined by \eqref{eq:Atilde_kappa_lmi}, and assume \(\tilde A(1)\) is Hurwitz. An interval \([\kappa_1,\kappa_2]\subset\R_{>0}\) is called \emph{LMI-certifiable} if there exists \(P=P^\top\succ 0\) such that \eqref{eq:lmi_A0}--\eqref{eq:lmi_kappa_high} hold with \((\underline\kappa,\bar\kappa)=(\kappa_1,\kappa_2)\).
Define
\begin{align}
\overline\kappa_{\rm LMI}
&:=\sup\{\kappa\ge 1:\ [1,\kappa]\ \text{is LMI-certifiable}\}, \label{eq:kappa_up_LMI} \\
\underline\kappa_{\rm LMI}
&:=\inf\{\kappa\in(0,1]:\ [\kappa,1]\ \text{is LMI-certifiable}\}. \label{eq:kappa_low_LMI}
\end{align}
The interval \([\underline\kappa_{\rm LMI},\overline\kappa_{\rm LMI}]\) is called the \emph{LMI-certified gain interval}.
\end{definition}

The next result shows that the interval in
Definition~\ref{def:lmi_gain_interval} is a conservative inner approximation of
the exact stability-preserving gain interval.

\begin{proposition}
\label{prop:inner_approx_GM}
Let \([\kappa_{\min},\kappa_{\max}]\) denote the stability-preserving gain
interval from Definition~\ref{def:gain_interval}. Then
\begin{equation}
[\underline\kappa_{\rm LMI},\overline\kappa_{\rm LMI}]
\subseteq
[\kappa_{\min},\kappa_{\max}].
\label{eq:inner_bounds}
\end{equation}
\end{proposition}

\begin{proof}
If \([1,\kappa]\) is LMI-certifiable, then by
Prop.~\ref{prop:lmi_gain_interval}, the matrix \(\tilde A(\kappa)\) is Hurwitz throughout that interval, so
\([1,\kappa]\subseteq[\kappa_{\min},\kappa_{\max}]\). Taking the supremum over
all such \(\kappa\ge 1\) gives \(\overline\kappa_{\rm LMI}\le\kappa_{\max}\).
The argument for \(\underline\kappa_{\rm LMI}\ge\kappa_{\min}\) is analogous on
\((0,1]\). Therefore, \eqref{eq:inner_bounds} holds. 
\end{proof}

The next lemma records the monotonicity property that allows the endpoints of
the certified interval to be computed by repeated feasibility tests on nested
intervals.

\begin{lemma}
\label{lem:monotone_bisection}
If an interval \([\kappa_1,\kappa_2]\subset\R_{>0}\) is LMI-certifiable, then
every subinterval \([\hat\kappa_1,\hat\kappa_2]\subseteq[\kappa_1,\kappa_2]\)
is also LMI-certifiable.
\end{lemma}

\begin{proof}
Let \(P\succ 0\) certify \([\kappa_1,\kappa_2]\). For any
\(\hat\kappa\in[\kappa_1,\kappa_2]\), write
$\hat\kappa=\mu\kappa_1+(1-\mu)\kappa_2$,
with $\mu\in[0,1]$.
Since \(\tilde A(\kappa)\) depends affinely on \(\kappa\),
\begin{align*}
\tilde A(\hat\kappa)^\top P+P\tilde A(\hat\kappa)
&=
\mu\bigl(\tilde A(\kappa_1)^\top P+P\tilde A(\kappa_1)\bigr) \\
&\quad +(1-\mu)\bigl(\tilde A(\kappa_2)^\top P+P\tilde A(\kappa_2)\bigr),
\end{align*}
which is negative definite because \(P\) certifies the endpoints. Hence the
same \(P\) certifies every subinterval.
\end{proof}

Lem.~\ref{lem:monotone_bisection} implies that the endpoints \(\underline\kappa_{\rm LMI}\) and \(\overline\kappa_{\rm LMI}\) can be computed by repeated feasibility tests on refined intervals using the LMIs in Prop.~\ref{prop:lmi_gain_interval}. In the next subsection, these LMIs are used directly as synthesis constraints to design controllers with prescribed gain robustness.

\subsection{Synthesis for Prescribed Gain Robustness}
\label{subsec:gain_synthesis}

We now turn from certification to synthesis. Here the objective is not merely
to verify robustness for a given controller, but to design the nominal feedback
gain \(K\) so that a desired gain interval \([\underline\kappa,\bar\kappa]\) guarantees that the origin is GES for the closed loop~\eqref{eq:twomode}. The key point is
that, for fixed \((\underline\kappa,\bar\kappa)\), the endpoint LMIs in
Prop.~\ref{prop:lmi_gain_interval} can be converted into convex
synthesis conditions by the standard change of variables \(X=P^{-1}\) and
\(Y=KX\).

Recall that
\begin{equation}
\tilde A(\kappa)=A_0-\kappa UV^\top,
\qquad
A_0=A-BK,
\label{eq:Atilde_kappa_synthesis}
\end{equation}
with \(U\) fixed by \(B\), \(c\), and \(G\). For synthesis, it is convenient to
rewrite \(\tilde A(\kappa)\) in a form affine in \(K\).

The next lemma gives this decomposition. The key observation is that, once
\(\kappa\), \(G\), and \(\alpha\) are fixed, the dependence of the perturbed
active mode on the gain \(K\) is affine.

\begin{lemma}
\label{lem:Atilde_affine_in_K}
Let \(\tilde A(\kappa)\) be given by \eqref{eq:Atilde_kappa_synthesis}. For each \(\kappa>0\), define
\begin{equation}\begin{aligned}
A_\kappa &:= (I-\kappa Uc^\top)A-\kappa\alpha Uc^\top, \\
B_\kappa &:= (I-\kappa Uc^\top)B.
\label{eq:ABk}
\end{aligned}\end{equation}
Then \(\tilde A(\kappa)=A_\kappa-B_\kappa K\).
\end{lemma}

\begin{proof}
Using \(A_0=A-BK\) and \(V^\top=c^\top(A_0+\alpha I)\), one has
\(
\tilde A(\kappa)
=
A-BK-\kappa Uc^\top(A-BK+\alpha I).
\)
Expanding terms gives
\(
\tilde A(\kappa)
=
\bigl((I-\kappa Uc^\top)A-\kappa\alpha Uc^\top\bigr)
-
\bigl((I-\kappa Uc^\top)B\bigr)K,
\)
which is exactly $A_{\kappa}-B_{\kappa}K$.
\end{proof}

Lem.~\ref{lem:Atilde_affine_in_K} allows the endpoint certificates in
Prop.~\ref{prop:lmi_gain_interval} to be written as convex LMIs in the
variables \(X\) and \(Y\).

The next proposition gives the corresponding synthesis condition. It shows that
a prescribed guaranteed gain interval can be enforced by solving a single
feasibility problem.

\begin{proposition}
\label{prop:gain_interval_synthesis}
Fix \([\underline\kappa,\bar\kappa]\subset\mathbb{R}_{>0}\) with \(\underline\kappa\le 1\le \bar\kappa\), and let \(A_\kappa\) and \(B_\kappa\) be defined as in \eqref{eq:ABk}. If there exist \(X=X^\top\succ 0\) and \(Y\) such that
\begin{equation}
AX+XA^\top-BY-Y^\top B^\top\prec 0,
\label{eq:lmi_nominal_synthesis}
\end{equation}
\begin{equation}
A_{\underline\kappa}X+XA_{\underline\kappa}^\top
-B_{\underline\kappa}Y-Y^\top B_{\underline\kappa}^\top\prec 0,
\label{eq:lmi_low_synthesis}
\end{equation}
and
\begin{equation}
A_{\bar\kappa}X+XA_{\bar\kappa}^\top
-B_{\bar\kappa}Y-Y^\top B_{\bar\kappa}^\top\prec 0,
\label{eq:lmi_high_synthesis}
\end{equation}
then the gain \(K=YX^{-1}\) guarantees that the origin is GES for the closed loop~\eqref{eq:twomode} for every \(\kappa\in[\underline\kappa,\bar\kappa]\).
\end{proposition}

\begin{proof}
Let \(P:=X^{-1}\). Premultiplying and postmultiplying
\eqref{eq:lmi_nominal_synthesis} by \(P\) gives
\[
(A-BK)^\top P+P(A-BK)\prec 0.
\]

Using Lem.~\ref{lem:Atilde_affine_in_K},
\eqref{eq:lmi_low_synthesis} and \eqref{eq:lmi_high_synthesis} yield
\[
\tilde A(\underline\kappa)^\top P+P\tilde A(\underline\kappa)\prec 0,
\qquad
\tilde A(\bar\kappa)^\top P+P\tilde A(\bar\kappa)\prec 0.
\]
Therefore, Prop.~\ref{prop:lmi_gain_interval} applies and implies that
\(\tilde A(\kappa)\) is Hurwitz for every
\(\kappa\in[\underline\kappa,\bar\kappa]\), and that the origin is GES for the closed loop~\eqref{eq:twomode}.
\end{proof}

Prop.~\ref{prop:gain_interval_synthesis} gives a convex design condition
for selecting \(K\) so that a desired gain interval is guaranteed to preserve
global exponential stability. It also allows one to search over
candidate interval endpoints and thereby design \(K\) to enlarge the certified
gain-robustness range of the safety filter. In practice, this can be done by fixing a parameterization of the interval
(e.g., \([\underline\kappa,\bar\kappa]=[1/\gamma,\gamma]\)) and solving the
LMIs repeatedly while increasing \(\gamma\) until feasibility is lost.

\subsection{Certified Phase- and Delay-Robustness Bounds}
\label{subsec:lmi_phase_delay}

The gain case is especially simple because the perturbation enters affinely in
\(\kappa\). Phase robustness can also be certified, but now through a
conservative sufficient condition. The phase-uncertainty set is an arc of the
unit circle, which is not directly convenient for LMI-based certification.
Enclosing this arc in a disk yields a conservative but tractable uncertainty
description, reducing the problem to robust stability with respect to a bounded
scalar complex perturbation.

\begin{lemma}
\label{lem:phase_disk_reduction}
Fix \(\bar\phi\in[0,\pi/2]\), and define $c_\phi:=\cos\bar\phi$, $r_\phi:=\sin\bar\phi$.
Then,
\begin{equation}
\{e^{-j\phi}:|\phi|\le \bar\phi\}
\subseteq
\{\Delta\in\mathbb C:|\Delta-c_\phi|\le r_\phi\}.
\label{eq:phase_disk_containment}
\end{equation}
\end{lemma}

\begin{proof}
Let \(\Delta=e^{-j\phi} = \cos\phi-j\sin\phi\) with \(|\phi|\le \bar\phi\).
Then,
\begin{align*}
|\Delta-c_\phi|^2
&=(\cos\phi-c_\phi)^2+\sin^2\phi =1-2c_\phi\cos\phi+c_\phi^2.
\end{align*}
Since \(|\phi|\le \bar\phi\) and \(\bar\phi\in[0,\pi/2]\), we have
\(\cos\phi\ge \cos\bar\phi=c_\phi\). Hence
\[
|\Delta-c_\phi|^2
\le 1-c_\phi^2
=\sin^2\bar\phi
=r_\phi^2,
\]
which proves \eqref{eq:phase_disk_containment}.
\end{proof}

Now, define 
\begin{equation}
A_\phi:=A_0-c_\phi UV^\top.
\label{eq:A_phi}
\end{equation}
For every complex scalar \(\Delta\) in the enclosing disk, one can write
$\Delta=c_\phi+\delta$, with $\delta\in\mathbb C$ satisfying $|\delta|\le r_\phi$.
Therefore
$\tilde A(\Delta)=A_\phi-\delta UV^\top$.
Thus the phase-certification problem reduces to robust stability with respect to
the bounded complex scalar perturbation \(\delta\).

The next result gives a tractable sufficient condition for certifying a phase
range. The key idea is to use the disk enclosure from
Lem.~\ref{lem:phase_disk_reduction}, rewrite the perturbed active mode as a feedback interconnection with a bounded complex scalar uncertainty, and then apply the bounded real lemma.

\begin{proposition}
\label{prop:lmi_phase_bound}
Let \(A_\phi\) be defined by \eqref{eq:A_phi}. Suppose there exists
\(P=P^\top\succ 0\) such that
\begin{equation}
\begin{bmatrix}
A_\phi^\top P+P A_\phi & P U & r_\phi V \\
U^\top P & -1 & 0 \\
r_\phi V^\top & 0 & -1
\end{bmatrix}
\prec 0.
\label{eq:lmi_phase}
\end{equation}
Then \(\tilde A(e^{-j\phi})\) is Hurwitz for every \(|\phi|\le \bar\phi\). Consequently, the origin is GES for the closed loop~\eqref{eq:twomode} for every \(|\phi|\le \bar\phi\).
\end{proposition}

\begin{proof}
By Lem.~\ref{lem:phase_disk_reduction} and the definition of $\Delta$ and $\tilde A(\Delta)$, the phase-certification problem reduces to robust stability of
\[
\tilde A(\Delta)=A_\phi-\delta UV^\top,
\qquad
|\delta|\le r_\phi,
\]
with respect to a bounded complex scalar perturbation \(\delta\).
This  is the closed-loop state matrix of the feedback system
\[
\dot z=A_\phi z+Uw,
\qquad
y=V^\top z,
\qquad
w=-\delta y.
\]
The corresponding transfer function from \(w\) to \(y\) is
\[
M_\phi(s)=V^\top(sI-A_\phi)^{-1}U.
\]
Moreover, since \eqref{eq:lmi_phase} is negative definite, its upper-left
principal block satisfies
\(
A_\phi^\top P+PA_\phi\prec 0.
\)
Because \(P=P^\top\succ 0\), Lyapunov's theorem implies that \(A_\phi\) is
Hurwitz, and hence \(M_\phi\) is stable. By the bounded real lemma,
\eqref{eq:lmi_phase} is equivalent to
\(
\|r_\phi M_\phi\|_\infty<1;
\)
see, e.g., \cite[Ch.~2]{ZhouDoyleGlover1996} or
\cite[Sec.~2.7]{BoydElGhaouiFeronBalakrishnan1994}. Therefore, by small gain,
\(\tilde A(\Delta)\) is Hurwitz for every \(|\delta|\le r_\phi\), and in
particular for every \(\Delta=e^{-j\phi}\) with \(|\phi|\le \bar\phi\).
Finally, since \(A_0\) is Hurwitz by assumption and \(\tilde A(e^{-j\phi})\) is
Hurwitz for all \(|\phi|\le \bar\phi\), Prop.~\ref{prop:ges_recall} implies
that the origin is GES for the corresponding closed loop~\eqref{eq:twomode}.
\end{proof}


The previous result gives a certified phase bound for the active loop. The corresponding delay bound follows from the phase lag \(\omega\tau\) induced by a pure delay.

Define the gain-crossover set
\begin{equation}
\Omega_{\rm gc}:=\{\omega\ge 0:|L(j\omega)|=1\},
\label{eq:Omega_gc}
\end{equation}
and let
\(
\omega_{\max}:=\sup \Omega_{\rm gc},
\)
assuming \(\Omega_{\rm gc}\) is nonempty and \(\omega_{\max}<\infty\).

The next result converts the certified phase bound into a certified delay bound
for the active loop.

\begin{corollary}
\label{cor:delay_from_phase}
Let \(\omega_{\max}\) be defined by \eqref{eq:Omega_gc}, and suppose that, for some \(\bar\phi\in[0,\pi/2]\), there exists \(P=P^\top\succ 0\) satisfying \eqref{eq:lmi_phase}. Then the delayed active-loop equation
\begin{equation}
1+e^{-s\tau}L(s)=0
\label{eq:delay_loop_eq}
\end{equation}
has no solutions in \(\{\Re s\ge 0\}\) for every \(0\le \tau\le \bar\phi/\omega_{\max}\).
\end{corollary}

\begin{proof}
For a delay factor \(e^{-s\tau}\), the phase lag at frequency \(\omega\) is
\(\omega\tau\), while the loop magnitude is unchanged. If \(\tau\) satisfies
$0\le \tau\le \frac{\bar\phi}{\omega_{\max}}$, then
\(
\omega\tau\le \bar\phi
\text{ for every }\omega\in\Omega_{\rm gc}.
\)
Thus the delay contributes no more phase lag than the certified phase bound at
any gain-crossover frequency. By the delay boundary condition associated with \eqref{eq:delay_loop_eq}, this
excludes imaginary-axis crossings of the delayed active loop. Hence the delayed active-loop equation has no solutions in \(\{\Re s\ge 0\}\)
for all such \(\tau\).
\end{proof}

Prop.~\ref{prop:lmi_phase_bound} and Cor.~\ref{cor:delay_from_phase}
provide conservative but tractable certified phase- and delay-robustness
bounds for the active loop. These conditions can also, in principle, be incorporated in controller design
in the same spirit as the gain case.

\section{Numerical Examples}

\subsection{Illustrative Planar Example}

We consider a planar illustrative example to demonstrate the certified robustness analysis and the gain-synthesis method.
The dynamics and safety constraint are given by
\[
A=\begin{bmatrix}
0 & 1\\
-2 & -0.5
\end{bmatrix}, \
B=\begin{bmatrix}
0\\
1
\end{bmatrix},
\
c=\begin{bmatrix}
0.5\\
1
\end{bmatrix},
\
d=1.
\]
We choose \(\alpha=1\), \(G=1\), and the nominal controller
\(
K=\begin{bmatrix}
0.2 & 0.2
\end{bmatrix}.
\)
For the baseline controller, the exact upper endpoint of the active-mode
gain interval, obtained from the frequency-domain characterization in
Cor.~\ref{cor:gain_interval_freq}, is
\(
\kappa_{\max}\approx 1.294,
\)
while the corresponding certified upper endpoint, computed from the LMIs in
Prop.~\ref{prop:lmi_gain_interval}, is
\(
\kappa_{\max}^{\rm cert}\approx 1.196.
\)
The exact phase and delay margins, obtained from the active-loop conditions
in \eqref{eq:phase_boundary}--\eqref{eq:delay_boundary}, are
$\phi_{\rm PM}\approx 33.8^\circ$, $\tau_{\rm DM}\approx 1.50~\text{s}$.
whereas the certified phase and delay bounds, computed from
Prop.~\ref{prop:lmi_phase_bound} and
Cor.~\ref{cor:delay_from_phase}, are
$\phi_{\rm PM}^{\rm cert}\approx 18.4^\circ$, $\tau_{\rm DM}^{\rm cert}\approx 0.159~\text{s}$.
%
%
For reference, the inactive mode \(A_0\) is much more robust: it remains stable for all positive uniform gain scalings and has no finite phase-margin limit, so the robustness bottleneck here is the active mode.


We then apply the gain-synthesis method of
Prop.~\ref{prop:gain_interval_synthesis} to enforce the desired
certified gain interval
\(
[1,\bar\kappa]=[1,2].
\)
The resulting controller is
\(
K_{\mathrm{syn}}=
\begin{bmatrix}
-1.4030 & -0.2230
\end{bmatrix},
\)
for which the exact and certified gain endpoints become $\kappa_{\max}\approx 6.155$, $\kappa_{\max}^{\rm cert}\approx 6.154$.
The exact phase and delay margins, again computed from
\eqref{eq:phase_boundary}--\eqref{eq:delay_boundary}, improve to $\phi_{\rm PM}\approx 90.8^\circ$, $\tau_{\rm DM}\approx 1.18~\text{s}$.
while the certified phase and delay bounds, obtained from
Prop.~\ref{prop:lmi_phase_bound} and
Cor.~\ref{cor:delay_from_phase}, become $\phi_{\rm PM}^{\rm cert}\approx 50.7^\circ$, $\tau_{\rm DM}^{\rm cert}\approx 0.562~\text{s}$.
%
%

Fig.~\ref{fig:quadrotor_gain_interval} compares the exact and certified gain intervals before and after synthesis. Fig.~\ref{fig:quadrotor_bode} shows the corresponding active-loop Bode plots and phase/delay margins. Fig.~\ref{fig:quadrotor_phase_portraits} shows the resulting state trajectories and confirms convergence in both cases.


\begin{figure}[t]
    \centering
    \includegraphics[width=.65\linewidth]{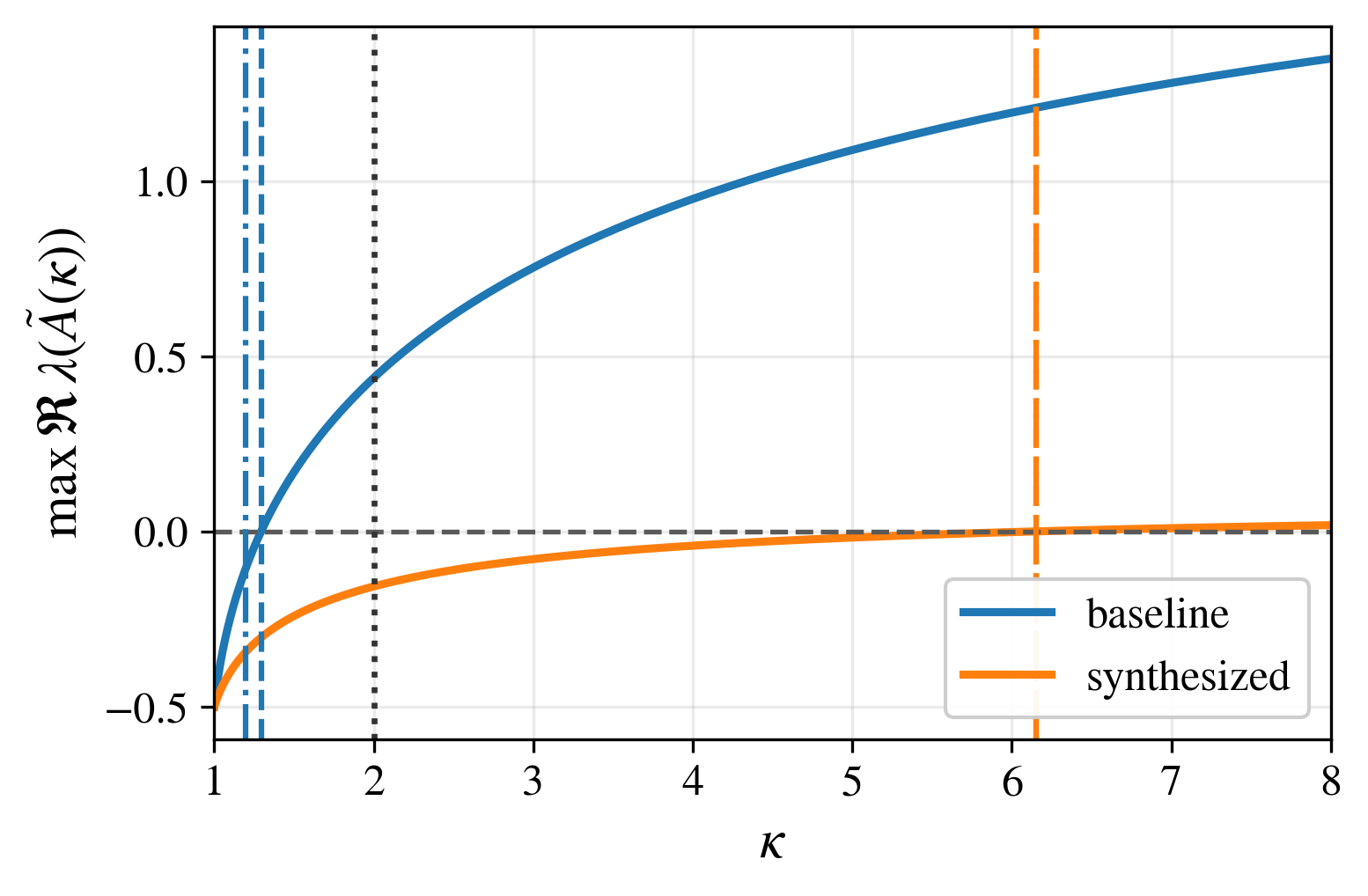}
    \caption{\footnotesize
    Spectral abscissa of \(\tilde A(\kappa)\) for the baseline and synthesized
    controllers. Dashed and dash-dotted lines denote exact and certified upper
    gain endpoints, respectively, and the dotted line marks the target
    \(\bar\kappa=2\).
    }
   
    \label{fig:quadrotor_gain_interval}
\end{figure}

\begin{figure}[t]
    \centering
    \includegraphics[width=.75\linewidth]{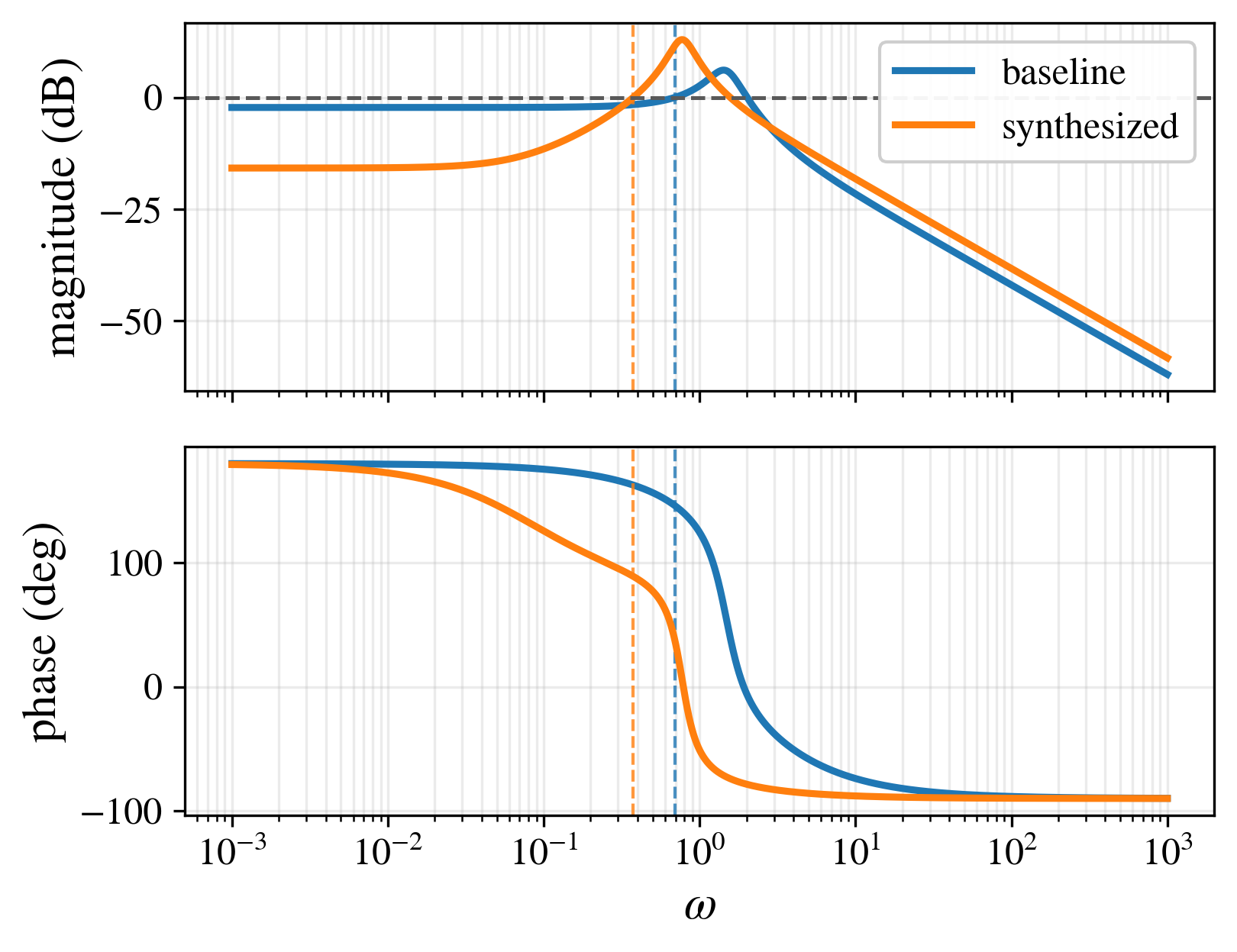}
    \caption{\footnotesize
    Active-loop Bode plots for the baseline and synthesized controllers.
    Dashed vertical lines mark the gain-crossover frequencies determining the
    minimum exact phase margin.
    }
    \label{fig:quadrotor_bode}
\end{figure}

\begin{figure}[t]
    \centering
    \includegraphics[width=\linewidth]{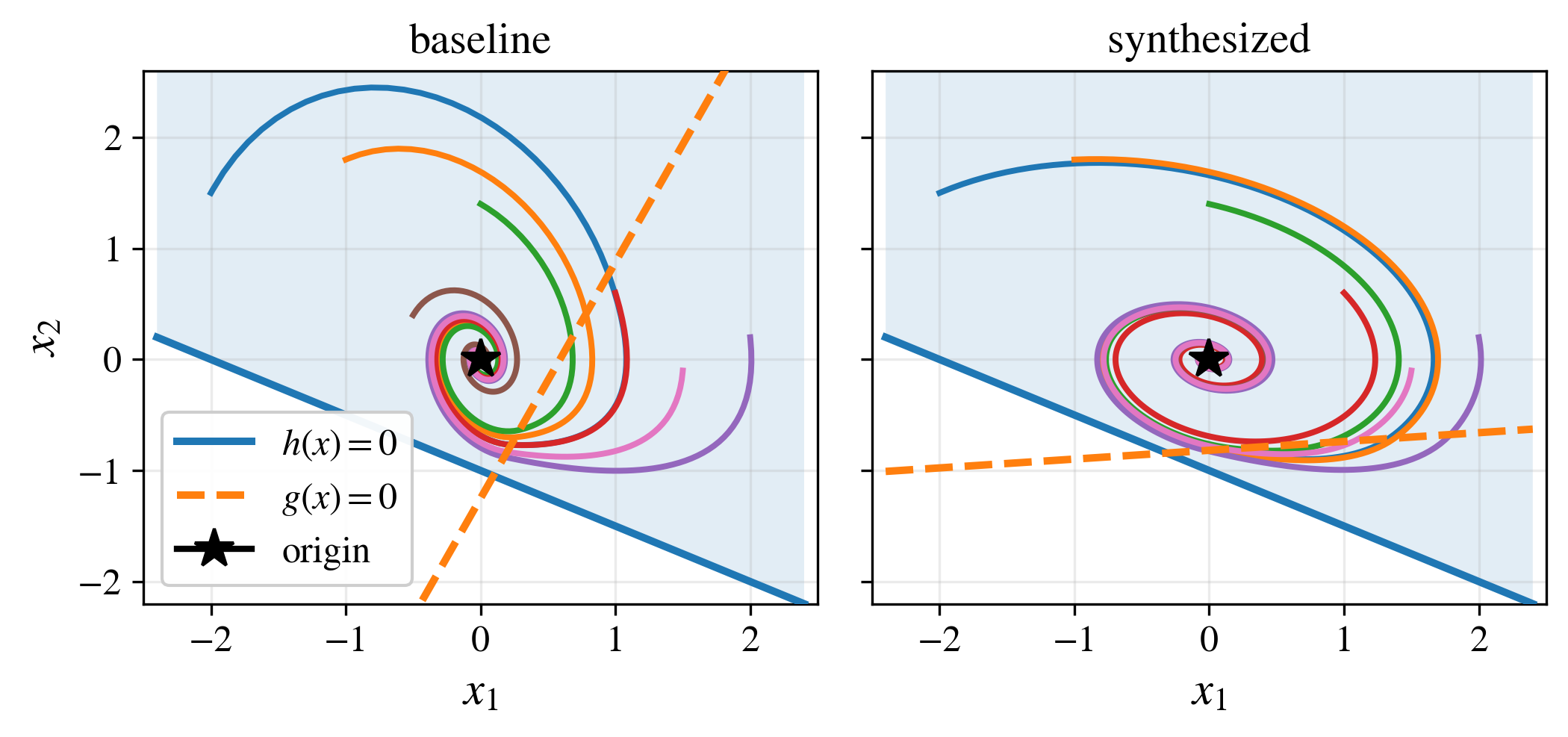}
    \caption{\footnotesize
    Phase portraits for the baseline and synthesized controllers. The solid and
    dashed lines denote \(h(x)=0\) and \(g(x)=0\), respectively.
    }
    \label{fig:quadrotor_phase_portraits}
\end{figure}

\subsection{Aircraft Roll--Yaw Example}

We next consider the roll--yaw dynamics of a mid-size aircraft around an
operating point from \cite[Sec.~14.8]{EL-KAW:24}. The plant state is
\(x_p=[\beta,p_s,r_s]^\top\), where \(\beta\) is the sideslip angle and
\(p_s,r_s\) are the roll and yaw rates. Using an integrated tracking-error state \(e_{yI}\), the augmented dynamics are
\begin{equation}
\begin{bmatrix}
\dot e_{yI}\\
\dot x_p
\end{bmatrix}
=
A
\begin{bmatrix}
e_{yI}\\
x_p
\end{bmatrix}
+
B u
+
\begin{bmatrix}
-y_{\rm cmd}\\
0
\end{bmatrix},
\label{eq:aircraft_augmented}
\end{equation}
with
\[
A=
\begin{bmatrix}
-\kappa & C_p\\
0 & A_p
\end{bmatrix},
\qquad
B=
\begin{bmatrix}
0\\
B_p
\end{bmatrix},
\qquad
C_p=[\,0\ \ 1\ \ 0\,],
\]
where \(\kappa=0.5\), and \(A_p,B_p\) are taken from \cite{EL-KAW:24}. We use
an LQR nominal controller with
\[
Q=\mathrm{diag}(10.25,0,0,16.02),
\qquad
R=\mathrm{diag}(1,0.49).
\]

Safety is imposed through the roll-rate constraint \(p_s\le 0.4\), which yields
the affine barrier function \(h(x)=c^\top x+d\) with
\(
c=\begin{bmatrix}0&0&-1&0\end{bmatrix}^\top,
\,
d=0.4.
\)
For the baseline controller, the exact upper endpoint of the active-mode gain
interval, computed from Cor.~\ref{cor:gain_interval_freq}, is
\(
\kappa_{\max}\approx 1.018.
\)
Moreover, the interval \([0.1,1.015]\) is certified by the LMI condition in
Prop.~\ref{prop:lmi_gain_interval}. The corresponding exact phase and
delay margins, obtained from \eqref{eq:phase_boundary}--\eqref{eq:delay_boundary},
are $\phi_{\rm PM}\approx 2.88^\circ$, $\tau_{\rm DM}\approx 1.09~\text{s}$.
%
%
For reference, the inactive mode \(A_0\) is stable for all positive uniform gain scalings and has phase margin approximately \(73^\circ\), so the robustness bottleneck here is the active mode.


Applying Prop.~\ref{prop:gain_interval_synthesis} to the target certified interval \([0.1,2]\) yields the gain
\[
K=
\begin{bmatrix}
-0.0040 & 0.1795 & 0.0521 & -0.1151\\
0.0002 & -2.0770 & -0.0130 & -0.2588
\end{bmatrix}.
\]
For this synthesized controller, the exact upper gain endpoint increases to
\(
\kappa_{\max}\approx 16.58.
\)
 
Fig.~\ref{fig:aircraft_gain_interval} shows the spectral abscissa of the
active-mode matrix \(\tilde A(\kappa)\) before and after synthesis.
Fig.~\ref{fig:aircraft_bode} shows the corresponding active-loop Bode plots.
Fig.~\ref{fig:aircraft_time_response} compares representative closed-loop
responses and confirms safe evolution under the filtered dynamics.

\begin{figure}[t]
    \centering
    \includegraphics[width=\linewidth]{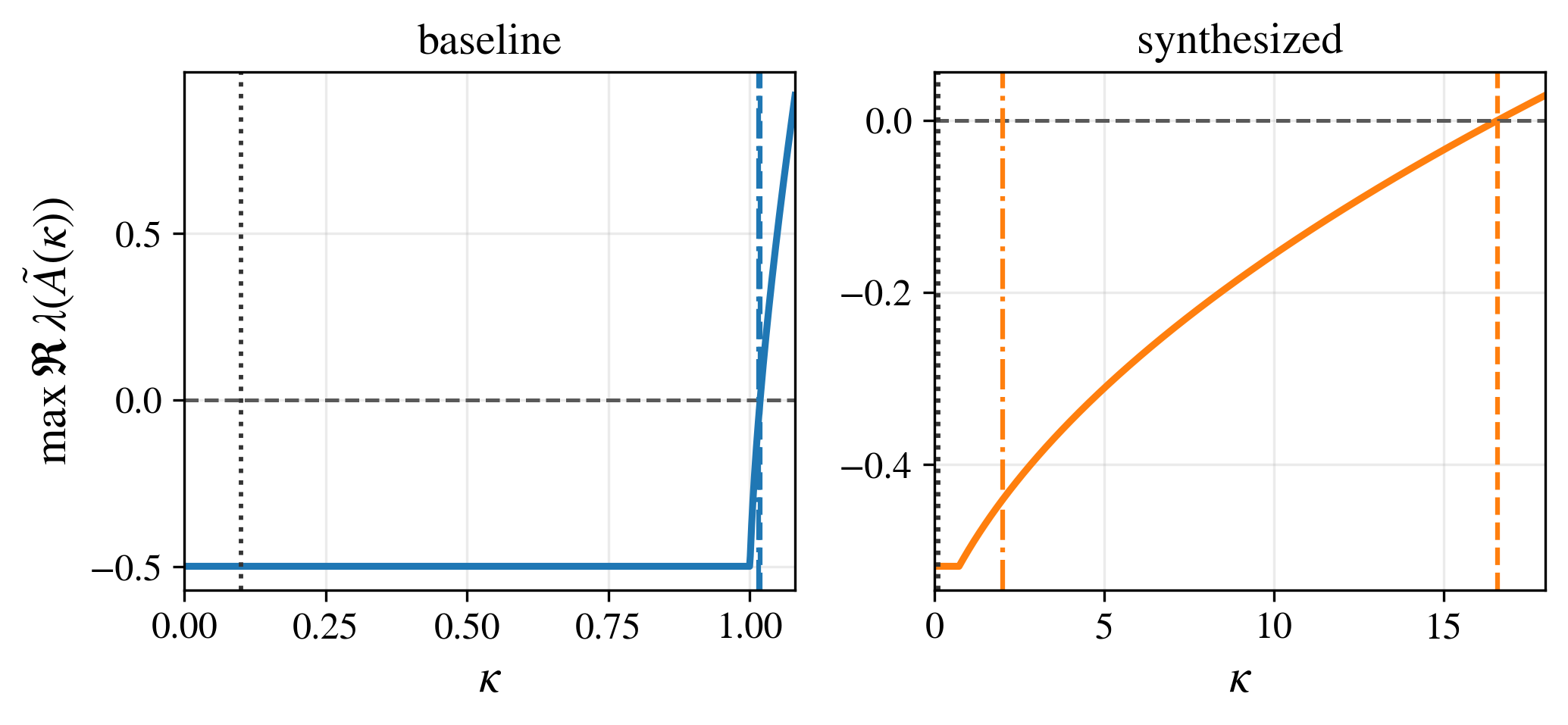}
    \caption{\footnotesize
    Spectral abscissa of \(\tilde A(\kappa)\) for the baseline and synthesized
    controllers. Dashed lines denote exact upper gain endpoints, and dotted
    lines denote the certified design targets.
    }
    \label{fig:aircraft_gain_interval}
\end{figure}

\begin{figure}[t]
    \centering
    \includegraphics[width=.75\linewidth]{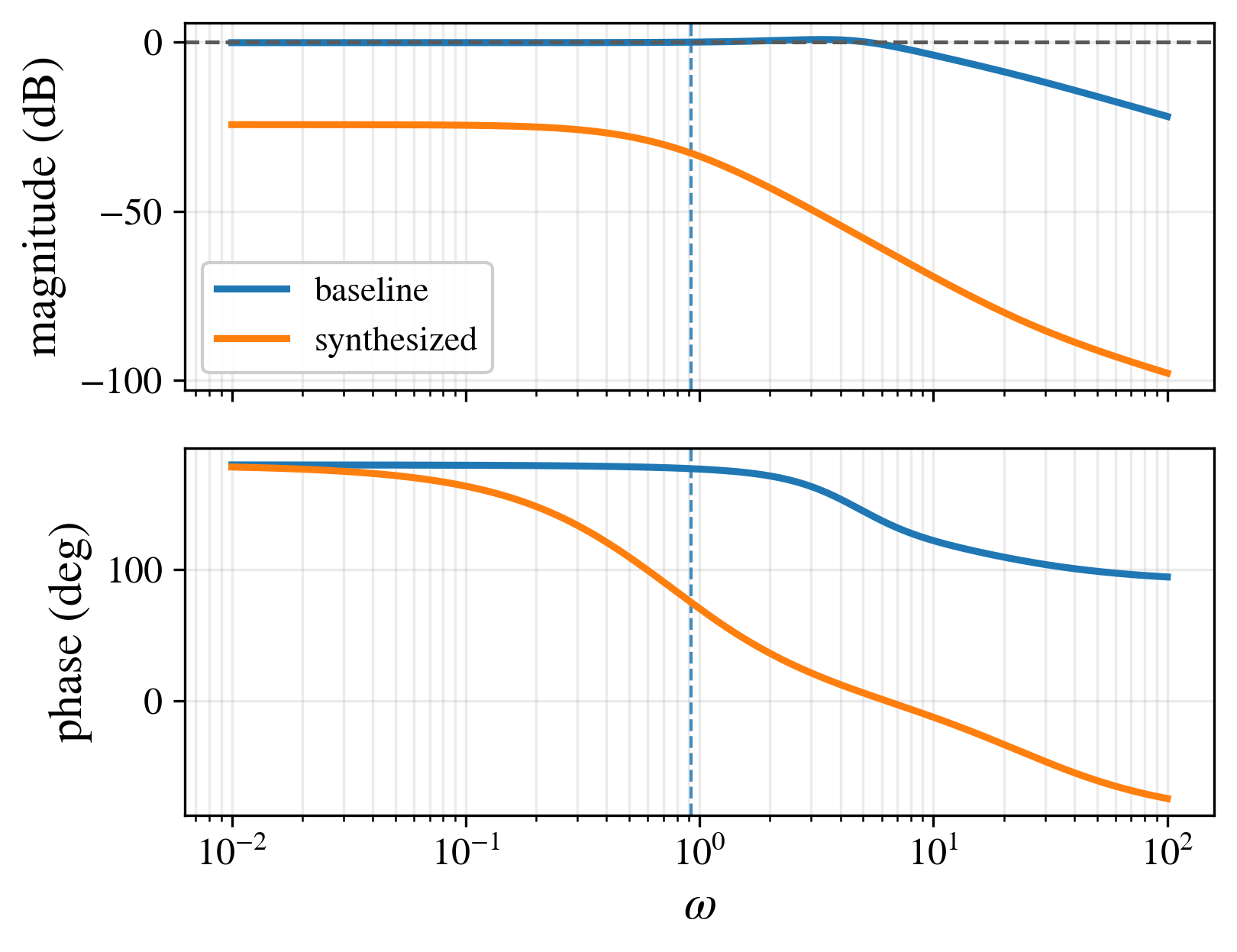}
    \caption{\footnotesize
    Active-loop Bode plots for the aircraft example. The baseline controller
    has a very small phase margin, while the synthesized controller yields a
    substantially more robust active loop.
    }
    \label{fig:aircraft_bode}
\end{figure}

\begin{figure}[t]
    \centering
    \includegraphics[width=\linewidth]{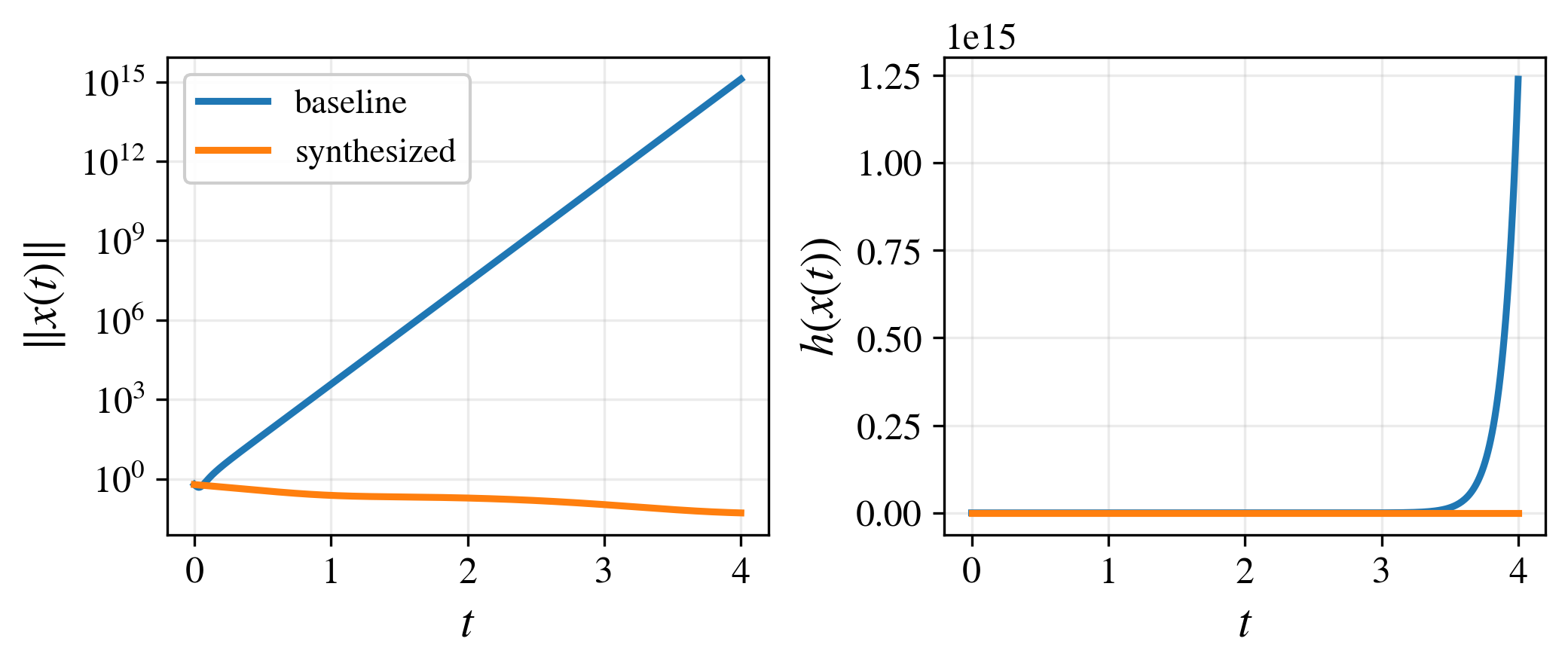}
    \caption{\footnotesize
    Representative safety-filtered closed-loop responses for the aircraft
    example, showing state evolution and barrier values for the baseline and
    synthesized controllers.
    }
    \label{fig:aircraft_time_response}
\end{figure}

\section{Conclusion}\label{sec:conclusion}

This paper studied certified robustness of closed-loop stability for CBF-QP safety filters with an affine safety constraint. The active mode was shown to admit an exact scalar loop representation, enabling gain, phase, and delay margin analysis from a classical robust-control viewpoint. This yielded exact stability-margin characterizations, as well as LMI-based certificates and synthesis conditions for prescribed robustness guarantees. Numerical examples illustrated the results.


\bibliography{bib/alias,bib/Main-add,bib/Main,bib/JC,bib/New,bib/PM, references}
\bibliographystyle{IEEEtran}

\end{document}